\documentclass[11pt]{article}
\usepackage{latexsym}
\usepackage[top=1in, bottom=1in, left=1in, right=1in]{geometry}

\usepackage{graphicx} 
\usepackage{color} 
\usepackage{amsmath, amsthm, amssymb}
\usepackage{enumerate}
\usepackage{float}
\usepackage{url}
\usepackage{stackrel}
\usepackage{mathrsfs,dsfont}
\usepackage{caption}
\usepackage{subcaption}
\usepackage{wrapfig}
\usepackage{bbm}
\usepackage{bm}
\usepackage{tikz}
\usepackage{pdfpages}

\newtheorem{theorem}{Theorem}[section]
\newtheorem{corollary}[theorem]{Corollary}
\newtheorem{proposition}[theorem]{Proposition}
\newtheorem{lemma}[theorem]{Lemma}

\theoremstyle{definition}
\newtheorem{example}[theorem]{Example}

\newtheorem{definition}[theorem]{Definition}

\newtheorem{algorithm}[theorem]{Algorithm}
\newtheorem{question}[theorem]{Question}

\theoremstyle{remark}
\newtheorem{remark}[theorem]{Remark}

\newcommand{\A}{U_{2345}}
\newcommand{\R}{\mathbb{R}}

\def\been{\begin{enumerate}}
\def\enen{\end{enumerate}}

\def\CC{\mathcal C}

\begin{document}

\title{Obstructions to convexity in neural codes}

\author{Caitlin Lienkaemper\footnote{Corresponding author; Department of Mathematics, Harvey Mudd College,
301 Platt Boulevard,
Claremont CA 91711-5901, USA; {\tt clienkaemper@g.hmc.edu}}, 
Anne Shiu\footnote{Department of Mathematics, Texas A\&M University, Mailstop 3368, College Station, Texas 77843--3368, USA; {\tt annejls@math.tamu.edu}}, 
and Zev Woodstock\footnote{Department of Mathematics and Statistics, James Madison University, Roop Hall 305, MSC 1911, Harrisonburg, Virginia 22807, USA; Present address: Department of Mathematics, Box 8205, NC State University, Raleigh, North Carolina 27695--8205
{\tt 
zwoodst@ncsu.edu} }
}

\date{December 18, 2016}

\maketitle

\abstract{
How does the brain encode spatial structure?  One way is through hippocampal neurons called place cells, which
become associated to convex regions of space known as their receptive fields: each place cell fires at a high rate precisely when the animal is in the receptive field. The firing patterns of multiple place cells form what is known as a convex neural code. How can we tell when a neural code is convex?  To address this question, Giusti and Itskov identified a local obstruction, defined via the topology of a code's simplicial complex, and proved that convex neural codes have no local obstructions. Curto {\em et al.}\ proved the converse for all neural codes on at most four neurons. Via a counterexample on five neurons, we show that this converse is false in general. Additionally, we classify all codes on five neurons with no local obstructions.  This classification is enabled by our enumeration of  connected simplicial complexes on 5 vertices up to isomorphism.  Finally, we examine how local obstructions are related to maximal codewords (maximal sets of neurons that co-fire).  Curto {\em et al.}\ proved that a code has no local obstructions if and only if it contains certain ``mandatory'' intersections of maximal codewords. We give a new criterion for an intersection of maximal codewords to be non-mandatory, and prove that it classifies all such non-mandatory codewords for codes on up to five  neurons.  
}
\vskip .1in

\noindent {\bf Keywords:} neural code, place cell, convex, good cover, simplicial complex, homology

\vskip .1in
\noindent {\bf MSC codes:} 
05E45,   
55U10,  
92C20   

\section{Introduction}

The brain's ability to navigate within and represent the physical world is fundamental to our everyday experience and ability to function.  How does the brain accomplish this?  For their work shedding light on this question, neuroscientists John O'Keefe, May Britt Moser, and Edvard Moser won the 2014 Nobel Prize in Physiology and Medicine. Their work led to the discovery of place cells, grid cells, and head direction cells, all of which take part in rodents' and other animals' mechanisms for representing, navigating through, and forming memories of their environments.  

This paper focuses on place cells, which are hippocampal neurons which become associated to regions of the environment known as their receptive fields or place fields. When an animal is located in a place cell's receptive field, the place cell fires at a higher rate than when the animal is outside the place field.  The firing patterns of a collection of place cells describe an animal's position within its environment. These receptive fields have been experimentally observed to be approximately convex regions of space. {\em Convex codes} are those neural codes (firing patterns) that can arise from the activity of place cells with convex receptive fields. 

Which neural codes are convex?  What are signatures of convexity or non-convexity?  Curto {\em et al.}~\cite{what-makes,neural_ring} and Giusti and Itskov~\cite{no-go} addressed these questions using combinatorial topology and commutative algebra, and gave complete answers for codes on up to four neurons.  
Curto {\em et al.} achieved this classification by organizing neural codes according to their simplicial complexes,  and, additionally, by focusing on {\em local obstructions} to convexity. 
Earlier, Giusti and Itskov had introduced this concept and proved that codes with local obstructions are necessarily non-convex. Curto {\em et al.}\ proved that local obstructions have the following interpretation: for each simplicial complex $\Delta$, there is a set of ``mandatory'' codewords whose presence in a code (whose simplicial complex is $\Delta$) is required to avoid local obstructions~\cite{what-makes}.  Therefore, a code must contain all its mandatory codewords to be convex.  Moreover, the mandatory codewords are necessarily intersections of maximal codewords.  This motivates the following questions:
\begin{question} \label{q:1}
Is every code which has no local obstructions convex?
\end{question}
\begin{question} \label{q:2}
Is every intersection of maximal codewords a mandatory codeword?
\end{question}
\begin{question} \label{q:3}
For codes on five neurons, which have local obstructions?  Which are convex?
\end{question}
\noindent
Our work addresses all three questions.  

In a preliminary version of~\cite{what-makes}, the answer to Question~\ref{q:1} was conjectured to be ``yes'', and this was verified for codes on up to four neurons. Moreover, this was the main open problem in this subject.  Here we demonstrate that even for codes on five  neurons, the answer is in fact ``no'': Theorem~\ref{thm:counterex} gives the first example of a non-convex code with no local obstructions.

For Question~\ref{q:2}, again the first negative answer appears in codes on five neurons~\cite{what-makes}.  Here we give a sufficient criterion for an intersection of maximal codewords to be non-mandatory (Theorem~\ref{thm:tree}).  Furthermore, our criterion classifies all such non-mandatory codewords for codes on five neurons. 
In other words, our result shows that codes with no local obstructions on at most five neurons are precisely those codes that contain all intersections of maximal codewords together with those codes that satisfy our new criterion.

Finally, we completely answer the first part of Question~\ref{q:3} by first enumerating the 157 connected simplicial complexes on five vertices, and then determining for each simplicial complex which codewords are mandatory.  Here we recall that a code has a local obstruction if and only if it is missing a mandatory codeword.  Our enumeration is therefore an important step toward answering the second part of the Question~\ref{q:3}, which we leave for future work.

\section{Background} \label{sec:background}
In this section, we introduce our assumptions, definitions, and notation. We approximate neural activity as binary: under this model, neurons are either firing or they are not. We encode the combinatorial data generated by the firing patterns of place cells as a neural code (Definition~\ref{def:code}). We index the neurons with the positive integers $\{1, \ldots, n\}=:[n]$. 

\subsection{Neural codes}
Biologically, a codeword corresponds to a set of neurons which fire together while no other neurons fire, and a neural code describes which groups of neurons are observed firing together:
\begin{definition} \label{def:code}
A {\em neural code} $\mathcal{C}$ on $n$ neurons is a set of subsets of $[n]$ (called {\em codewords}), i.e.\ $\mathcal{C} \subseteq 2^{[n]}$.  A {\em maximal} codeword in $\mathcal{C}$ is a codeword that is not properly contained in any other codeword in $\mathcal{C}$. 
\end{definition}

\begin{definition} \label{def:realize}
For a neural code $\mathcal C$ on $n$ neurons, a collection $\mathcal U = \{U_1,U_2, \ldots, U_n\}$ of subsets of a set $X$ \emph{realizes} $\mathcal C$ if a codeword $\sigma$ is in $\mathcal C$ if and only if $\left( \bigcap_{i\in\sigma} U_i \right) \setminus \bigcup_{i\notin\sigma}{U_i}$ is nonempty. \end{definition}
\noindent
In this paper we make the simplifying assumption that $X \supsetneq \bigcup_{i \in [n]} U_i$, i.e.\ the empty set is a codeword in every code. 

\begin{definition} \label{def:code-ppties}
A neural code is: 
\begin{enumerate}
	\item {\em intersection-complete} if it is closed under taking intersections. 
	\item {\em max-intersection-complete} if it is closed under taking intersections of maximal codewords. 
	\item a {\em good-cover} code if it can be realized by a good cover $\mathcal U = \{U_1, U_2, \ldots, U_n\}$ of some set $X \subseteq \mathbb{R}^d$.  (Recall that $\mathcal{U}$ is a {\em good cover} of $X$ if the $U_i$'s are contractible open sets that cover $X$ and each intersection $U_{i_1} \cap U_{i_2} \cap \cdots \cap U_{i_k}$ is contractible or empty.)
	\item \emph{convex} if it can be realized by a collection of convex open sets $U_1, U_2, \ldots, U_n \subseteq \R^d$. The \emph{minimal embedding dimension} of a code is the smallest value of $d$ for which this is possible.  
\end{enumerate}
\end{definition}

\begin{example}
Consider the code $\mathcal C = \{\{1,2\}, \{1,3\}, \{2,3\}, \{1\}, \{2\}, \{3\}, \emptyset \}$. We will generally write this as $\CC = \{  \mathbf{12,13,23},1,2,3\}$, where the maximal codewords are marked in bold. 
We interpret this code to mean that each pair of neurons fires together and each neuron fires alone, but all three neurons never fire at the same time. This code is intersection-complete, max-intersection-complete, a good-cover code, and convex, as realized here:
 \begin{figure}[h]
 \centering
 \includegraphics[width = 2.5 in]{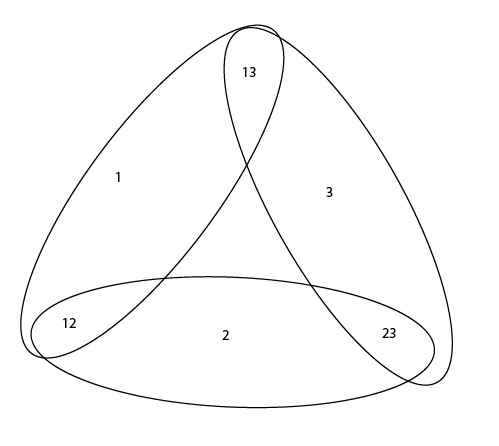}
 \end{figure}
\end{example}

\begin{example} \label{ex:non-convex}
 The simplest example of a neural code that is not convex is the three-neuron code $\CC=\{\mathbf{12,13}, \emptyset \}$. To see this, suppose $\CC$ were convex; then there would exist convex open sets $U_1, U_2$, and $U_3$ such that $U_1=U_2\cup U_3$, and $U_2\cap U_3=\emptyset$. We see that $U_2$ and $U_3$ would form a disconnection of the open set $U_1$, thus $U_1, U_2, $ and $U_3$ cannot all be convex open sets. In fact, since we did not use convexity in this argument, but only connectedness, we have shown that $\mathcal C$ is not a good-cover code either. 
\end{example}

\begin{remark}
Intersection patterns of convex (and other types of) sets is a well-established subject; for instance, see~\cite{helly, kalai, intersection-gph-book, tancer-survey} and the references therein.  Nevertheless, the related questions we consider here---which focus on {\em all} regions cut out by the convex sets, in addition to which sets intersect---have only recently received attention.
\end{remark}

\subsection{Simplicial complexes}
An abstract {\em simplicial complex} on $n$ vertices is a nonempty set of subsets ({\em faces}) of $[n]$ that is closed under taking subsets. 
(Thus, the empty set is an element of every simplicial complex.) That is, if $\Delta$ is a simplicial complex, then $\sigma\in\Delta$ and $\tau\subset\sigma$ implies $\tau\in \Delta$.   {\em Facets} are the faces of a simplicial complex that are maximal with respect to inclusion.  

For a code $\mathcal{C}$ on $n$ neurons, $\Delta(\mathcal{C})$ is the smallest simplicial complex on $[n]$ that contains $\mathcal{C}$:
\begin{align*}
	\Delta(\mathcal{C}) ~:=~ \{\omega \subseteq [n] \mid \omega \subseteq \sigma {\rm ~for~some~} \sigma \in \mathcal{C} \}~.
\end{align*}
Note that two codes on $n$ neurons have the same simplicial complex $\Delta$ if and only if they have the same maximal codewords (which are the facets of $\Delta$). 
For a face $\sigma \in \Delta$, the {\em link of $\sigma$ in $\Delta$} is the simplicial complex
\begin{align*}
{\rm Lk}_{\Delta}(\sigma) ~:=~ \{\omega \in \Delta \mid \sigma \cap \omega = \emptyset, \ \ \sigma \cup \omega \in \Delta \}~.
\end{align*}  
Next, the {\em restriction} of $\Delta$ to $\sigma$ is the simplicial complex
\begin{align*}
\Delta|_{\sigma} ~:= \{\omega \in \Delta \mid \omega \subseteq \sigma \}~.
\end{align*}
 Finally, a simplicial complex is {\em contractible} if its geometric realization is contractible.  

\subsection{Local obstructions} \label{sec:loc-obs}
Here we introduce local obstructions, which prevent a code from being convex, and furthermore prevent a code from being a good-cover code (Proposition~\ref{prop:prior-results}).
\begin{definition} \label{def:local-obs}
Let $\mathcal{C}$ be a code on $n$ neurons, let $\Delta=\Delta(\CC)$, and let $\mathcal{U}=\{U_1,U_2,\dots, U_n\}$ be any collection of open sets that realizes $\CC$.  The code $\CC$ has a \emph{local obstruction} if there exist disjoint, nonempty sets $\sigma, \tau \subseteq [n]$ such that:
	\begin{enumerate}
	\item $\left(\cap_{i \in \sigma} U_i \right) \cap U_j$ is nonempty for all $j \in \tau$, 
    \item $\left( \cap_{i \in \sigma} U_i \right)  ~\subseteq~ \left( \cup_{j \in \tau} U_j \right)$, and
    \item ${\rm Lk}_{\Delta|_{\sigma \cup \tau}} (\sigma) $ is not contractible.
	\end{enumerate}
\end{definition}
\noindent
It is important to note that the definition of local obstruction does not depend on the choice of realization $\mathcal{U}$.  Indeed, this can be seen from the following characterization of codes with local obstructions, due to Curto~{\em et al.}~\cite[Theorem 1.6]{what-makes}:
\begin{proposition}[Characterization of codes with local obstructions via maximal codewords] \label{prop:loc-obs}
A neural code $\mathcal{C}$ has a \emph{local obstruction} if and only if some nonempty intersection of maximal codewords is not in $\mathcal{C}$ and has a non-contractible link. More precisely, if $\mathcal{M} = \{M_1, ..., M_m\}$ is the set of maximal codewords of $\CC$, we say that $\mathcal{C}$ has a \emph{local obstruction} if for some $I \subseteq [m]$,
 \begin{enumerate}
 \item $\sigma : = \bigcap_{i \in I} M_i$ is nonempty,
 \item $\sigma \notin \mathcal{C}$, and
 \item ${\rm Lk}_{\Delta(\mathcal{C})}(\sigma)$ is not contractible.
 \end{enumerate}
\end{proposition}

Two observations follow immediately from Proposition~\ref{prop:loc-obs}.  First, each simplicial complex  defines a set of {\em mandatory} codewords, those nonempty intersections of facets for which the link is non-contractible:
\begin{definition} \label{def:mandatory}
A face $\sigma$ of a simplicial complex $\Delta$ is a {\em mandatory codeword} of $\Delta$ if it is the nonempty intersection of a set of facets of $\Delta$ such that ${\rm Lk}_{\Delta}(\sigma)$ is non-contractible.  
\end{definition}
\noindent
Proposition~\ref{prop:loc-obs} then states that
{\em a code $\CC$ has no local obstructions if and only if it contains all of the  mandatory codewords of $\Delta(\mathcal C)$.}

The second observation is that Proposition~\ref{prop:loc-obs} gives a method for determining which codes on $n$ neurons have local obstructions.  Namely, first enumerate all simplicial complexes on $n$ vertices, and then for each simplicial complex, determine the set of mandatory codewords.  Curto {\em et al.} completed this analysis for $n \leq 4$ (and additionally proved that all such codes without local obstructions are convex)~\cite{what-makes}, and we will complete the $n=5$ case in Section~\ref{sec:enum}.

\begin{example}
We revisit the non-convex three-neuron code $\CC= \{\mathbf{12,13}, \emptyset\}$ from Example~\ref{ex:non-convex}.
Its simplicial complex is $\Delta(\CC)=\{\mathbf{12,13},1,2,3, \emptyset\}$, a path of length 2. Thus, the codeword $1$ is an intersection of maximal codewords. However, ${\rm Lk}_{\Delta(\CC)}(1)=\{\mathbf{2,3}, \emptyset\}$, which is a simplicial complex consisting of two disconnected points. Thus, 1 is a mandatory codeword which is not in $\CC$, so $\CC$ has a local obstruction (by Proposition~\ref{prop:loc-obs}).
\end{example}

The following result summarizes prior results about the relationships among properties of codes; part~1 is due to Cruz~{\em et al.}~\cite{intersection-complete}, part~2 follows from the fact that every convex open cover is a good cover, part~3 is due to Giusti and Itskov~\cite{no-go} (see also~\cite[Lemma 1.5]{what-makes}), and part~4 follows immediately from Proposition~\ref{prop:loc-obs}.

\begin{proposition}	\label{prop:prior-results}
Let $\mathcal{C}$ be a neural code.
	\begin{enumerate}
	\item If $\CC$ is intersection-complete, then $\CC$ is convex.
	\item If $\CC$ is convex, then $\CC$ is a good-cover code.
	\item If $\CC$ is a good-cover code, then $\CC$ has no local obstructions. 
	\item If $\CC$ is max-intersection-complete, then $\CC$ has no local obstructions. 
	\end{enumerate}
\end{proposition}

\noindent
Curto~{\em et al.} showed that for all neural codes on up to four  neurons, the converses of parts 2--4 of Proposition~\ref{prop:prior-results} hold~\cite[\S 4]{what-makes}:
\begin{proposition}	\label{prop:at-most-4}  
Let $\CC$ be a neural code on at most four neurons. The following are equivalent:
	\begin{enumerate}
	\item $\CC$ is convex.
    \item $\CC$ is a good-cover code.
    \item $\CC$ has no local obstructions.
	\item $\CC$ is max-intersection-complete.
	\end{enumerate}
\end{proposition}

A preliminary version of~\cite{what-makes} conjectured that the equivalence of parts~1-3 in Proposition~\ref{prop:at-most-4} generalizes to codes on more than four neurons.  In the next section, we give the first counterexample to that conjecture, which shows that parts 1 and 2 are not equivalent (Theorem~\ref{thm:counterex}).  Furthermore, our counterexample uses only five neurons.  It is still unknown whether parts 2 and 3 are equivalent. 

Part~4 of Proposition~\ref{prop:prior-results} also can not be generalized to codes with more than four neurons~\cite{what-makes}.
Section~\ref{sec:tree} focuses on this gap: the intersections of facets that are {\em not} mandatory.  More precisely, we identify a criterion that guarantees that a max-intersection-incomplete code has no local obstructions, and we prove that our criterion classifies all such codes on up to five neurons.

\begin{example} \label{ex:counterexample}
Consider the following code: $$\CC = \{{\bf 2345,~123,~134, ~145},~13,~14,~23,~34,~ 45,~ 3,~ 4 ,~\emptyset\}~.$$ 
Note that $\CC$ is invariant under the permutation $(2,5)(3,4)$.  

The nonempty intersections of maximal codewords are $13$, $14$, $23$, $34$, $45$, $1$, $3$, and $4$. Of these, only $1$ is missing from the code. We find that 
$${\rm Lk}_{\Delta(\CC)}(1)=\{\mathbf{23,34,45},2,3,4,5, \emptyset\}~, $$ 
which is a path of length 3 (thus, contractible). Hence, $\CC$ has no local obstructions (Proposition~\ref{prop:loc-obs}), and, moreover, is a max-intersection-{\em incomplete} code with no local obstructions.   
Also, it is the minimal code (with respect to inclusion) among all codes with no local obstructions with the same simplicial complex (in the notation introduced later, $\mathcal{C} = \CC_{\rm min} (\Delta(\mathcal{C})) $ -- see Section \ref{sec:path}); this is because the links with respect to all other maximal codewords are non-contractible.  

In the next section, we show that despite having no local obstructions, $\CC$ is not convex. 
\end{example}

\section{A non-convex code with no local obstructions} \label{sec:counterex}

Here we show that the code in Example~\ref{ex:counterexample} is non-convex despite having no local obstructions (recall from Proposition~\ref{prop:prior-results} that good-cover codes have no local obstructions):
	\begin{theorem}	\label{thm:counterex}
	The following code on $5$ neurons:
	\begin{align} \label{eq:counterex-code}
	\mathcal{C} \quad = \quad 
		\left\{ {\bf 2345,~123,~134, ~145},~13,~14,~23,~34,~ 45,~ 3,~ 4,~ \emptyset \right\}
	\end{align}
	is a non-convex, good-cover code. 
	\end{theorem}

The proof of Theorem~\ref{thm:counterex} requires the following lemma, which has proven useful in other contexts as well (for instance, in distinguishing between minimal embedding dimension 2 vs.\ 3 among 
convex codes arising from the simplicial complex labeled by L24 in~\cite{what-makes}). 

\begin{lemma} \label{lem:cvx-sets-R2}
Let $W_1$, $W_2$, and $W_3$ be convex open sets in $\mathbb{R}^n$ such that their intersection is nonempty and is equal to all pairwise intersections: $W_1 \cap W_2 \cap W_3 = W_i \cap W_j$ for all $1 \leq i < j \leq 3$.  Then, any line that intersects each of the $W_i$'s must intersect $W_1 \cap W_2 \cap W_3$.
\end{lemma}

\begin{proof}
Let $W_{123}:=W_1 \cap W_2 \cap W_3$.  
Assume for contradiction that there exists a line $L$ that intersects each $W_i \setminus W_{123}$ (for $i=1,2,3$) but {\em not} $W_{123}$.  For $i=1,2,3$, let $p_i$ be such an intersection point on $W_i$, i.e.\ $p_i \in L \cap \left( W_i \setminus W_{123} \right)$.  By relabeling if necessary, we may assume that $p_2$ lies between $p_1$ and $p_3$ on the line $L$.  Now let $p_{123} \in W_{123}$; then $p_{123}$ is {\em not} on $L$ by hypothesis.  So, our points have a triangular structure: 
\begin{center}
\includegraphics[scale=0.26]{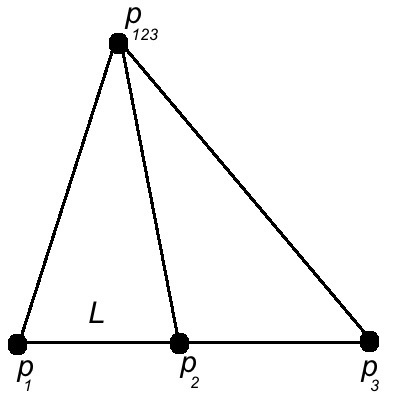}
\end{center}

Removing the convex open set $W_{123}$ from the line segment $\overline{p_2 p_{123}}$ yields a shorter, closed line segment (nonempty because it contains $p_2$), namely, the line segment $\overline{p_2 q} := \overline{p_2 p_{123}} \setminus W_{123}$ for some point $q$.  

We claim that the interior of the triangle $\bigtriangleup p_1  p_{123} q$ is contained in $W_1$.  Indeed, any point in the interior resides on a line segment between $p_1$ and a point (``above $q$'') on the half-open line segment $\overline{p_{123} q} \setminus \{q\}$, both of which are in the convex set $W_1$.  Similarly, the interior of $\bigtriangleup p_3  p_{123} q$ is contained in $W_3$.  

Let $L'$ be the line parallel to $L$ that passes through $q$.  Note that on one side of $q$ (the left side in the figure), $L'$ passes through the triangle $\bigtriangleup p_1  p_{123} q$, and on the other (right) side, $L'$ passes through $\bigtriangleup p_3  p_{123} q$.  Now $q$ is in the open set $W_2$, so there is an open neighborhood $N$ of $q$ that is entirely in $W_2$.  Intersect $N$ with $L'$ and pick from the ``left'' side a point $l \in N \cap L'$ which is therefore in $W_1 \cap W_2 = W_{123}$.  Similarly, pick a ``right'' point $r \in N \cap L'$ which is in $W_{123}$.  
\begin{center}
\includegraphics[scale=0.26]{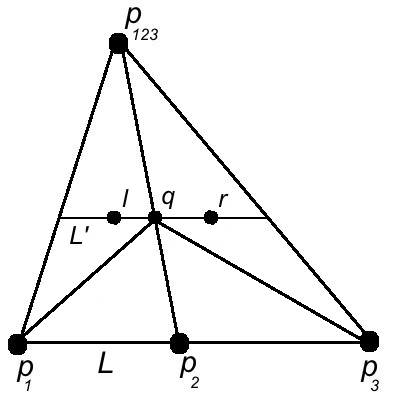}
\end{center}
Now, $q$ is on the line segment $\overline{lr}$ by construction, so $q \in W_{123}$ by convexity, which contradicts the construction of $q$.  Thus we are done.
\end{proof}

\begin{proof}[Proof of Theorem~\ref{thm:counterex}]
Figure~\ref{fig:good-cover} demonstrates that the code~\eqref{eq:counterex-code} is a good-cover code. 
 \begin{figure}[ht]
\begin{center}
\includegraphics[width = 1.7in]{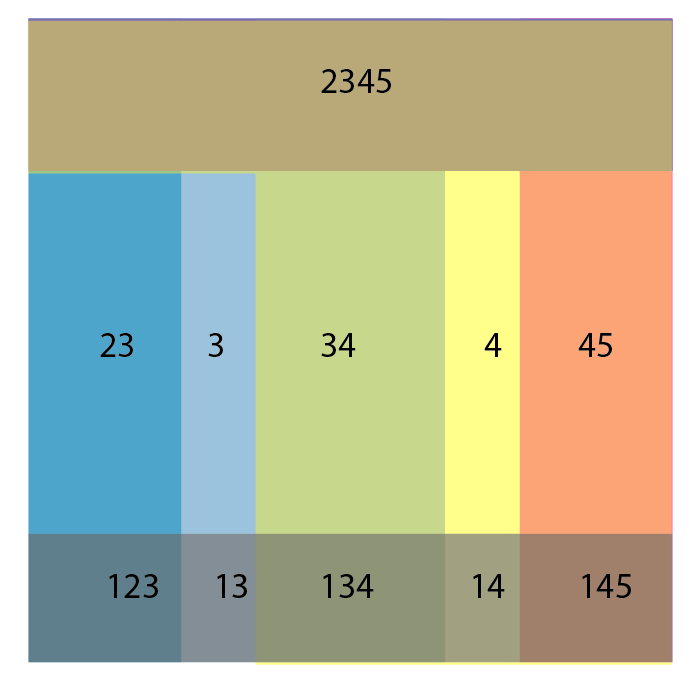}
\end{center}
\caption{A good-cover realization of the code~\eqref{eq:counterex-code}.  
More precisely, $U_i$ is the open set formed by the union of all regions in the figure that are labeled by a codeword that contains $i$.  For instance, $U_5$ is the union of the rectangular regions $2345$, $45$, and $145$.  It is straightforward to check that $\{U_i\}$ forms a good cover. \label{fig:good-cover}}
\end{figure}

Thus, it remains only to show that the code~\eqref{eq:counterex-code} is not convex.  
Assume for contradiction that there exist convex, open sets $U_1, U_2, \dots, U_5$ in $\mathbb{R}^d$ which realize the code $\mathcal{C}$.  Our strategy is to derive a contradiction with Lemma~\ref{lem:cvx-sets-R2}.

For $\sigma \subseteq [5]$, we define $U_{\sigma}:= \cap_{i \in \sigma} U_i$.  Let $p_{123} \in U_{123}$, 
$p_{145} \in U_{145}$, and
$p_{2345} \in U_{2345}$; these three points exist and are distinct, because 123, 134, and 2345 are maximal codewords of $\mathcal{C}$ (so, $U_{123}$ does not intersect $U_{145}$, and so on).  Also, we may assume that these three points are not collinear: if they were, since the $U_{\sigma}$'s are open sets, one of the points can be perturbed slightly.  

We claim that the line segment $L=\overline{p_{123} p_{145}}$ intersects $U_{134}$.  Indeed, by convexity, $L$ is contained in $U_1$, which is covered by $U_3$ and $U_4$, so $L$ is too.  The only way for $L$ (a connected set) to be covered by two open sets is if the sets overlap and this overlap intersects $L$, i.e.\ $L \cap U_3 \cap U_4 \neq \emptyset$.  Note that $L \cap U_3 \cap U_4 \subseteq U_{134}$, so we are done.
Let $p_{134} $ be a point in that intersection; thus, $p_{134}$ is in $U_{134}$.  So far, our points have the following configuration:
\begin{center}
\includegraphics[scale=0.26]{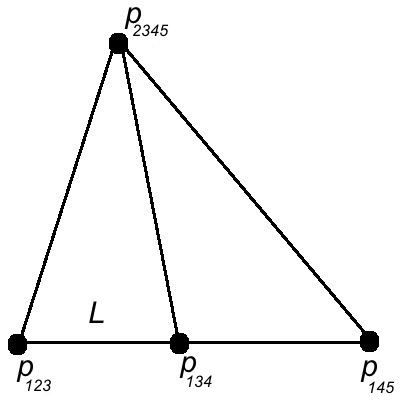}
\end{center}

With Lemma~\ref{lem:cvx-sets-R2} in mind, we define
	\begin{align*} 
	W_1:=U_2 \cap U_3~, \quad \quad 
	W_2:=U_3 \cap U_4~, \quad \quad 
	W_3:=U_4 \cap U_5~. 
	\end{align*}

First, note that the line extending the line segment  $L=\overline{p_{123} p_{145}}$ contains 
$p_{123} \in W_1$,  
$p_{134} \in W_2$,
 and $p_{145} \in W_3$. 
Also, we claim that this line does not intersect the triple-intersection.  Indeed, if the line contains a point $r$ in the triple-intersection $W_1 \cap W_2 \cap W_3 =\A $, 
then $r$ can not be on the line segment $L=\overline{p_{123} p_{145}}$, because $L$ is in $U_1$ and 2345 is a maximal codeword of $\mathcal{C}$.  Also, $r$ can not lie to the ``left'' of $p_{123}$, because then $p_{123} \in \overline{p_{134} r} \subseteq U_5$, which is a contradiction (123 is a maximal codeword of $\mathcal{C}$).  Similarly, $r$ can not lie to the ``right'' of $p_{134}$.

So, to reach a contradiction with Lemma~\ref{lem:cvx-sets-R2}, we need only check that the sets $W_i$ satisfy the hypotheses of Lemma~\ref{lem:cvx-sets-R2}.  
First, as intersections of convex open sets, 
the $W_i$'s are convex open sets in $\mathbb{R}^d$.  
Next, the triple-intersection $W_1 \cap W_2 \cap W_3 = \A$ 
contains the point $p_{2345}$, so is nonempty.  Finally, to show that the double-intersections coincide with the triple-intersection, it suffices to show that $W_i \cap W_j \subseteq \A$ for all $1 \leq i < j \leq 3$ ($\supseteq$ holds by construction). 
First, $W_1 \cap W_3 =\A$ 
by construction.  Next, $W_1 \cap W_2 = U_2 \cap U_3 \cap U_4 
\subseteq\A$, 
because $2345$ is the only codeword in $\mathcal{C}$ that contains $234$.  Analogously, $W_2 \cap W_3 \subseteq \A$, so we are done. 
\end{proof}

The proof of Theorem~\ref{thm:counterex} shows that there is a ``second-level'' obstruction: for codes having the same simplicial complex as the counterexample code~\eqref{eq:counterex-code}, if the codeword 1 is not in the code, then both 234 and 345 must be in the code (for the code to be convex).  Here we see that adding these two codewords does make the code convex:

\begin{proposition} \label{prop:add-2}
The following neural code (obtained by adding the codewords 234 and 345 to the counterexample code~\eqref{eq:counterex-code}) is convex, and thus a good-cover code:
\begin{align*}
		\left\{ {\bf 2345,~123,~134, ~145},~234,~345,~13,~14,~23,~34,~ 45,~ 3,~ 4, ~\emptyset \right\}~.
\end{align*}
\end{proposition}
\begin{proof}
First note that the following is a convex realization for the code obtained by adding codewords 234, 345, 2, and 5 to the code~\eqref{eq:counterex-code}:
\begin{center}
	\begin{tikzpicture}[scale=.8]
    \draw[rotate=15] (0, 0) rectangle (4, 1);
    \draw[rotate=30] (0, 0) rectangle (4, 1);
    \draw[rotate=45] (0, 0) rectangle (4, 1);
    \draw[rotate=60] (0, 0) rectangle (4, 1);
	\draw [gray,line width=12] (1.5,2.8) -- (2.5,1.5);
    \node [above] at (1.9,3.6) {$U_2$};
    \node [above] at (2.9,2.9) {$U_3$};
    \node [below] at (3.6,2.8) {$U_4$};
    \node [below] at (4.1,1.9) {$U_5$};
    \node [below] at (2.8,2.45) {$U_1$};
	\end{tikzpicture}
\end{center}
Now it is straightforward to see that if $U_2$ is replaced by $U_2 \cap U_3$ and $U_5$ is replaced by $U_4 \cap U_5$, then the resulting sets $U_i$ would be a convex realization that verifies our claim.
\end{proof}

Another way to make the counterexample code~\eqref{eq:counterex-code} convex, is simply to add the codeword $1$, which makes the code max-intersection-complete (Proposition~\ref{prop:add-1} below).  
(In contrast, the code in Proposition~\ref{prop:add-2} is max-intersection-{\em incomplete}.) 
Then, from the proof of~Theorem~\ref{thm:counterex}, a line from $p_{123}$ to $p_{145}$ no longer must pass through a point $p_{134}\in U_{134}$. Thus, we do not have the same forced structure which we used to show that $\CC$ is not convex. 
        
\begin{proposition} \label{prop:add-1}
The following neural code (obtained by adding the codeword $1$ to the counterexample code~\eqref{eq:counterex-code}) is convex, and thus a good-cover code:
\begin{align*}
		\left\{ {\bf 2345,~123,~134, ~145},~13,~14,~23,~34,~ 45,~1,~ 3,~ 4, ~\emptyset \right\}~.
\end{align*}
\end{proposition}
\begin{proof}
Consider the following construction. Let $U_{2345}$ be an open cube in $\R^3$, centered at the origin with sides parallel to the $x$, $y,$ and $z$ axes. Then let $U_{23}$ be a rectangular prism created by extending the cube in the positive $x$ direction, let $U_{34}$ be a rectangular prism created by extending the cube in the positive $y$ direction, and let $U_{45}$ be a rectangular prism created by extending the cube in the positive $z$ direction. Then let $U_3$ be the convex hull of $U_{23}$ and $U_{34}$, and let $U_4$ be the convex hull of $U_{34}$ and $U_{45}$. Thus, we have the codewords 2345, 23, 34, 45, 3, and 4. Now, pick points $p_{23}\in U_{23}$, $p_{34}\in U_{34}$, and $p_{45}\in U_{45}$, and let $U_1$ be an open $\epsilon$-neighborhood of the convex hull of $p_{23}$, $p_{34}$, and $p_{45}$. This creates regions corresponding to the codewords 1, 123, 134, and 145. This also must create regions corresponding to the codewords 13 and 14, since a line from 123 to 134 must pass through 13 and a line from 134 to 145 must pass through 14. It is straightforward to check that for $\epsilon$ sufficiently small, all the above codewords remain, and no new ones are created.  Therefore, this is a convex realization of the code.
\end{proof}

\begin{remark} \label{rmk:2-minimal}
The simplicial complex of the counterexample code $\mathcal{C}$ from~\eqref{eq:counterex-code}) is the 
first example of a simplicial complex $\Delta$ which has more than one minimal convex code with simplicial complex equal to $\Delta$.  Namely, both $\mathcal{C} \cup \{234, ~345\}$ and $\mathcal{C} \cup \{1 \}$ are minimal among all convex codes with simplicial complex equal to $\Delta(\mathcal{C})$.
\end{remark}

\section{Enumerating and classifying codes on five  neurons} \label{sec:enum}
Previous classification of neural codes on up to four neurons as having local obstructions or not was done by hand~\cite{what-makes,neural_ring}. We automated this process for codes on five neurons using {\tt SageMath} \cite{sage}: we first enumerated all simplicial complexes on five vertices, up to symmetry, and then computed for each simplicial complex the list of mandatory codewords.  We describe these procedures in more detail below.  

Our source code is available on {\tt SageMathCloud} at:
\\
\url{https://cloud.sagemath.com/projects/8fdd3fd5-5b65-4059-8e3f-95e02b104e84/files/obstructions_to_convexity_in_neural_codes.sagews}

\noindent
The list of simplical complexes and their mandatory codewords can be found in the Appendix.

\subsection{Enumerating simplicial complexes}
For our list of simplicial complexes, we first used {\tt Nauty} \cite{nauty} to generate a list of all connected simplicial complexes on up to five vertices, up to isomorphism. We worked only with connected simplicial complexes because any disconnected simplicial complex on five neurons can be expressed as the disjoint union of simplicial complexes on fewer than five vertices, and has thus been dealt with in previous work. We found that there is 
one connected simplicial complex on each of one and two vertices,  
three connected simplicial complexes on three vertices, 
14 connected simplicial complexes on  four vertices, and
157 connected simplicial complexes on five vertices. The simplicial complexes on up to four vertices appears in~\cite[Figure~4]{what-makes}. 

Including disconnected simplicial complexes brings these counts to 1, 2, 5, 20, and 180, respectively, which agrees with the corresponding sequence in the on-line encyclopedia of integer sequences~\cite[A261005]{oeis}.
The next term in this sequence tells us that there are 16,143 
simplicial complexes on six vertices, up to isomorphism. We do not produce a list of all simplicial complexes on six vertices, since we view this as too large a data set to be useful at the moment.

\subsection{Checking for (homological) local obstructions and enumerating mandatory codewords}
Our algorithm for computing the mandatory codewords of a simplicial complex follows closely the characterization local obstructions via maximal codewords (Proposition~\ref{prop:loc-obs}). It is important to recognize that Algorithm~\ref{alg} is not an algorithm in the strict sense. Eventually, it will fail (for some simplicial complexes) because having trivial homology groups does not imply contractibility in general, but does for simplicial complexes on few vertices.
\begin{algorithm}[Heuristic algorithm for computing mandatory codewords] \label{alg}

~\\
\noindent
{\em Input:} a simplicial complex $\Delta$ 

\noindent
{\em Output:} the list of mandatory codewords of $\Delta$

\noindent
{\em Initialize:} {\sc Mandatory}$:= \emptyset$

\noindent
{\em Steps:}
\begin{enumerate}
\item List all nonempty intersections of facets of $\Delta$.
\item Compute the link of each nonempty intersection of facets.
\item Compute the reduced homology groups of each link.
\item For each reduced homology group which is nontrivial, add the corresponding intersection of facets to {\sc Mandatory}.
\item Return {\sc Mandatory}. 
\end{enumerate}
\end{algorithm}

\section{The tree criterion for max-intersection-incomplete codes} \label{sec:tree}
Recall that mandatory codewords of a simplicial complex necessarily are intersections of facets, but not vice-versa.  
In this section, we present a new criterion for an intersection of facets to be non-mandatory (Theorem~\ref{thm:tree}), and then show that this criterion characterizes all intersection-{\em incomplete} codes without local obstructions for codes on at most five neurons (Theorem~\ref{thm:classify-5}).  However, this criterion is insufficient for codes on $6$ or more neurons (Example~\ref{ex:6}).  We end the section by showing that certain codes satisfying our new criterion are in fact convex with minimal embedding dimension~1 (Proposition~\ref{prop:path-cvx}).  

We begin with several definitions.  First, recall that for a finite collection $\mathcal{W}=\{W_1, W_2, \dots, W_n\}$ of subsets of a set $X$, the {\em nerve} of $\mathcal{W}$ is the simplicial complex that records the intersection patterns among the sets:
	\begin{align*}
	\mathcal{N}(\mathcal{W}) ~:=~ \left\{ I \subseteq [n] \mid \bigcap_{i \in I} W_i {\rm ~is~nonempty} \right\}~.
	\end{align*}
\noindent
Next, for a face $\sigma$ of a simplicial complex $\Delta$, we let $\mathcal{M}_{\Delta}(\sigma)$ denote the set of all facets (maximal faces) of $\Delta$ that contain $\sigma$.  Thus, if $\Delta=\Delta(\CC)$ for some code $\CC$, then $\mathcal{M}_{\Delta}(\sigma)$ is the set of all maximal codewords of $\CC$ that contain the codeword $\sigma$.
Finally, we let $\mathcal{L}_{\Delta} (\sigma)$ denote the set of facets of ${\rm Lk}_{\Delta}(\sigma)$; it is straightforward to see that these facets are obtained by removing $\sigma$ from the facets of $\Delta$ that contain $\sigma$:
	\begin{align*}
	\mathcal{L}_{\Delta} (\sigma) ~=~ \{ (M \setminus \sigma) \mid M \in \mathcal{M}_\Delta (\sigma) \}~.
    \end{align*}

We will need the following version of the nerve lemma~\cite[Theorem 6 and Remark 7]{Bjorner}:
\begin{lemma} \label{lem:nerve}
Let $\mathcal{D}=(\Delta_i)_{i \in I}$ be a family of sub-complexes of a connected simplicial complex $\Delta$ for which:
	\begin{enumerate}[(1)]
    \item $\Delta = \bigcup\limits_{i \in I} \Delta_i$, and 
    \item every finite nonempty intersection $\Delta_{i_1} \cap \Delta_{i_2} \cap \cdots \cap \Delta_{i_k}$ is contractible.
	\end{enumerate}
Then $\Delta$ is homotopy-equivalent to the nerve of the $\Delta_i$'s:
$	\Delta \simeq \mathcal{N}(\mathcal{D}).$
\end{lemma}

We apply Lemma~\ref{lem:nerve} in the following setting: 
the simplicial complex $\Delta$ is a link ${\rm Lk}_{\Delta}(\sigma)$,
and $\mathcal{D}$ is 
the set obtained from $\mathcal{L}_{\Delta} (\sigma)$ -- the set of facets $M \setminus \sigma$ of the link --
by replacing each $M \setminus \sigma$ by its simplicial complex $\Delta(\{ M \setminus \sigma\})$. 
Note that $\mathcal{D}$ and  $\mathcal{L}_{\Delta} (\sigma)$ have the same nerve:
$\mathcal{N}(\mathcal{D})= \mathcal{N}(\mathcal{L}_{\Delta} (\sigma))$.  
Next, condition (1) of the lemma holds, because the union of the facets is equal to the link.  
Condition (2) also holds, because nonempty intersections of any faces of a simplicial complex are themselves faces, which are contractible.  Therefore, the lemma and the equality of nerves mentioned above together imply that:
\begin{align} \label{eq:homotopy-eq}
{\rm Lk}_{\Delta}(\sigma)
~\simeq~
\mathcal{N} \left( \mathcal{L}_{\Delta} (\sigma) \right) ~.
\end{align}

\begin{theorem}[Tree criterion for mandatory codewords] \label{thm:tree}
Let $\Delta$ be a simplicial complex, and 
let $\sigma \in \Delta$ be a nonempty intersection of facets of $\Delta$.
If the nerve $\mathcal{N} \left( \mathcal{L}_{\Delta} (\sigma) \right) $ is a tree graph, then $\sigma$ is not a mandatory codeword of $\Delta$, i.e.\ ${\rm Lk}_{\Delta}(\sigma)$ is contractible.
\end{theorem}

\begin{proof}
The theorem follows directly from the definition of mandatory (Definition~\ref{def:mandatory}), the homotopy-equivalence~\eqref{eq:homotopy-eq}, and the fact that tree graphs are contractible.

\end{proof}

\begin{remark}
The nerve $\mathcal{N} \left( \mathcal{L}_{\Delta} (\sigma) \right) $ is a tree graph if and only if all triple-wise intersections among facets of $\Delta$ that contain $\sigma$ (i.e. elements of $\mathcal{M}_{\Delta}(\sigma)$) are equal to $\sigma$.
\end{remark}

What Theorem~\ref{thm:tree} says is that in light of the characterization of local obstructions in terms of intersections of facets (Proposition~\ref{prop:loc-obs}), those codewords that satisfy the tree criterion do not generate local obstructions.  We rephrase this in Corollary~\ref{cor:tree} below via the following definition:

\begin{definition} \label{def:tree}
A code $\CC$ on $n$ neurons {\em satisfies the tree criterion} if for every codeword $\sigma \subseteq [n]$ that is not in $\CC$ and is a nonempty intersection of maximal codewords of $\CC$, the nerve $\mathcal{N}\left(\mathcal{L}_{\Delta} (\sigma) \right)$ is a tree graph.
\end{definition}

\begin{corollary} \label{cor:tree}
If a neural code $\CC$ satisfies the tree criterion, then $\CC$ has no local obstructions.
\end{corollary}

\begin{example}
Returning again to our counterexample code~\eqref{eq:counterex-code},
recall from Example~\ref{ex:counterexample} that $1$ is the only nonempty intersection of maximal codewords that is not in the code, and that its link ${\rm Lk}_{\Delta(\CC)}(1) = \{\mathbf{23,34,45},2,3,4,5\}$ is a path of length 3.  Therefore, the set of facets of the link is $\mathcal{L}_{\Delta} (\sigma)= \{ {23,34,45} \}$, and hence the nerve $\mathcal{N}\left(\mathcal{L}_{\Delta} (\sigma) \right)$ is a path of length 2. We conclude that the code satisfies the tree criterion.  (We already knew that it has no local obstructions.)
\end{example}

\subsection{Using the tree criterion to classify codes with no local obstructions}
Max-intersection-complete codes (Definition~\ref{def:code-ppties}) vacuously satisfy the tree criterion (by Proposition~\ref{prop:loc-obs}).
Also, recall that max-intersection-complete codes on at most four neurons are precisely the codes with no local obstructions; in fact, they are convex (by Proposition~\ref{prop:at-most-4}). Thus, the converse of Corollary~\ref{cor:tree} is true for up to four neurons.  In fact, our next result extends this converse to five neurons.  However, this converse is false for codes on six or more neurons (Example~\ref{ex:6}). 

\begin{theorem}[Tree criterion characterization of codes with no local obstructions on up to five neurons] \label{thm:classify-5}
For a code $\CC$ on at most five neurons, $\CC$ has no local obstructions
if and only if
$\CC$ satisfies the tree criterion.
\end{theorem}
\begin{proof}[Proof of Theorem~\ref{thm:classify-5}]

\noindent
By Corollary~\ref{cor:tree}, a neural code which satisfies the tree criterion has no local obstructions. Our proof that the converse holds on up to five neurons requires the next lemma, due to Curto {\em et al.}~\cite[Lemma~2.11 and subsequent discussion]{what-makes}.  Recall that a simplicial complex $\Delta$ on $[n]$ is a {\em cone} if there exists $i \in [n]$ (called a {\em cone point}) such that every nonempty facet of $\Delta$ contains $i$. For instance, the simplicial complex $\{\{123\}, \{34\}\}$ is a cone with cone point 3, whereas the simplicial complex $\{\{123\}, \{34\}, \{45\}\}$ is not a cone, since there is no point contained in every facet. 

\begin{lemma} \label{lem:cone}
A simplicial complex $\Gamma$ is not a cone if and only if
there exists a simplicial complex $\Delta$ and a face $\sigma$ of $\Delta$ such that (1) $\sigma$ is an intersection of facets of $\Delta$ and (2) ${\rm Lk}_{\Delta}(\sigma)= \Gamma$.
\end{lemma}


Recall codes which satisfy the tree criterion have no local obstructions by Corollary~\ref{cor:tree}.  For the converse, let $\mathcal{C}$ be a code on $n$ neurons, where $n \leq 5$, that has no local obstructions.  Let $\sigma \subseteq [n]$ be a codeword that is not in $\mathcal{C}$ and is a nonempty intersection of maximal codewords of $\mathcal{C}$.  Then ${\rm Lk}_{\Delta(\CC)}(\sigma)$ is (1) a simplicial complex on at most four vertices (because $\sigma$ is nonempty), (2) contractible (by Proposition~\ref{prop:loc-obs}: $\mathcal{C}$ has no local obstructions) and thus nonempty, and (3) not a cone (by Lemma~\ref{lem:cone}).  Among the 28 simplicial complexes on up to four vertices (depicted in~\cite[Figure~4]{what-makes}), only one is a contractible non-cone: $P_3$, the path of length $3$.  Thus,  ${\rm Lk}_{\Delta(\CC)}(\sigma) \cong P_2$ (the path of length 2), which is a tree graph.  Hence, by definition, $\mathcal{C}$ satisfies the tree criterion. 
\end{proof}
%

\begin{example}[A convex code that the tree criterion misses] \label{ex:6}
Let 
\begin{align*}
\mathcal{C}~=~ \{ {\bf 124},~{\bf 134},~{\bf 145},~{\bf 156},~14,~15,~ \emptyset\}~,
\end{align*}
so $\Delta(\mathcal{C})$ is the cone, with cone point 1, over the following graph: 

\begin{center}
	\begin{tikzpicture}[scale=.5]
	\draw (0,0) --(1,1);
	\draw (0,2) --(1,1);
	\draw (2.5,1) --(1,1);
	\draw (2.5,1) --(4,1);
    \draw [fill] (0,0) circle [radius=0.08];
    \draw [fill] (1,1) circle [radius=0.08];
    \draw [fill] (0,2) circle [radius=0.08];
    \draw [fill] (2.5,1) circle [radius=0.08];
    \draw [fill] (4,1) circle [radius=0.08];
    \node [left] at (0,2) {$2$};
    \node [left] at (0,0) {$3$};
    \node [below] at (1,1) {$4$};
    \node [below] at (2.5,1) {$5$};
    \node [below] at (4,1) {$6$};
	\end{tikzpicture}
\end{center}
\noindent    
We claim that $\mathcal{C}$ has no local obstructions.  Indeed, the codeword $1$ is the unique nonempty intersection of maximal codewords of $\mathcal{C}$ that is not in $\mathcal{C}$, and ${\rm Lk}_{\Delta(\CC)}(1)$ is the above tree graph, which is contractible.

To complete the proof, we now show that $\mathcal{C}$ does not satisfy the tree criterion.  To see this, note that $\sigma=\{1\}$ is not in $\mathcal{C}$ and is the nonempty intersection of all maximal codewords of $\mathcal{C}$.  However, the nerve $\mathcal{N}\left(\mathcal{L}_{\Delta(\CC)} (\sigma) \right)$ is the following non-tree simplicial complex:

\begin{center}
	\begin{tikzpicture}[scale=.5]
    \draw [fill=gray, thick] (0,0) --(1,1) --(0,2) -- (0,0);
	\draw (2.5,1) --(1,1);
    \draw [fill] (0,0) circle [radius=0.08];
    \draw [fill] (1,1) circle [radius=0.08];
    \draw [fill] (0,2) circle [radius=0.08];
    \draw [fill] (2.5,1) circle [radius=0.08];
    \node [left] at (0,2) {$34$};
    \node [left] at (0,0) {$45$};
    \node [below right] at (1,1) {$24$};
    \node [below right] at (2.5,1) {$56$};
	\end{tikzpicture}
\end{center}
\noindent
This shows that $\mathcal{C}$ fails to satisfy the tree criterion.
\end{example}

\subsection{Proving convexity for a special case when the nerve is a path} \label{sec:path}
Here we show that a certain family of codes that satisfy the tree criterion is in fact convex with minimal embedding dimension 1.  The corresponding simplicial complexes only contain one non-mandatory codeword, and the corresponding nerve is the simplest type of tree: a path.  These simplicial complexes include cones over a path:

\begin{definition} \label{def:cone-path}
Let $\mathcal{M}$ denote the set of facets of a simplicial complex $\Delta$.  We say that $\Delta$ is a {\em coned path} if:
\begin{enumerate}
	\item the intersection of all facets is nonempty:  $\sigma:= \bigcap\limits_{M \in \mathcal{M}} M  \neq \emptyset$, and 
	\item the nerve $\mathcal{N}(\mathcal{L}_{\Delta}(\sigma))$ is a path graph (of length at least 1).
\end{enumerate}
\end{definition}

Next we will show that for each simplicial complex, there is a unique minimal code with no local obstructions, and then for the case of coned paths, prove that these minimal codes are convex with minimal embedding 1 (Proposition~\ref{prop:path-cvx}).  
To introduce minimal codes, recall from  Proposition~\ref{prop:loc-obs} that for a given simplicial complex $\Delta$, a code $\mathcal{C}$ with $\Delta(\CC) = \Delta$ has no local obstructions if and only if it contains the following code:
\begin{align*}
	\CC_{\rm min} (\Delta)~:=~ 
		\{ {\rm mandatory~codewords~of~} \Delta(\mathcal{C}) \} ~\cup~ \{ {\rm facets~of~} \Delta \} ~\cup~ \{\emptyset\} ~.
\end{align*}
In other words, $\CC_{\rm min} (\Delta)$ is the minimal code with simplicial complex $\Delta$ that has no local obstructions.
For instance, our counterexample code~\eqref{eq:counterex-code} is the minimal code for its simplicial complex.  
Now we determine this minimal code for the case of a coned path:
\begin{lemma} \label{lem:min-code}
Let $\sigma$ be the (nonempty) intersection of all the facets of a coned path $\Delta$.  Then:
\begin{enumerate}
	\item there exists an ordering of the facets of ${\rm Lk}_{\Delta}(\sigma)$:
	    $$\mathcal{L}_{\Delta}(\sigma) ~=~\{ N_1, N_2, \dots, N_m\}$$
	such that $N_i \cap N_{i+1} =: \tau_i$ is nonempty and all other pairwise intersections are empty: $N_i \cap N_j = \emptyset$ if $|i-j|\geq 2$, and
\item $\{ {\rm mandatory~codewords~of~}\Delta\}=\{ \sigma \cup \tau_i \mid 1 \leq i \leq m-1\}$.
\end{enumerate}
\end{lemma}

\begin{proof}
Part 1 follows immediately from the fact that the nerve $\mathcal{N}(\mathcal{L}(\sigma))$ is a path graph.
For part 2, we begin by noting that $\mathcal{M}=\{\sigma \cup N_i \mid 1 \leq i \leq m \}$ is the set of facets of $\Delta$.  By part 1, the only nonempty intersections of facets are $\sigma$ and the $\sigma \cup \tau_i$'s, and $\sigma$ is non-mandatory by the tree criterion (Theorem~\ref{thm:tree}).  So, we need only show that ${\rm Lk}_{\Delta}(\sigma \cup \tau_i)$ is non-contractible (by Proposition~\ref{prop:loc-obs}).  Indeed, we will see that this link is disconnected. By part 1, the only facets that contain $\sigma \cup \tau_i$ are $\sigma \cup N_i$ and $\sigma \cup N_{i+1}$, and by definition neither facet contains the other.  Thus ${\rm Lk}_{\Delta}(\sigma \cup \tau_i)$ is the disjoint union of two full simplices: one on $N_i \setminus \tau_i$ and one on $N_{i+1}\setminus \tau_i$., and thus is disconnected.

\end{proof}

\begin{proposition} \label{prop:path-cvx}
If $\Delta$ is a coned path, then the minimal code 
$\CC_{\rm min} (\Delta)$ 
is
convex with minimal embedding dimension 1.
\end{proposition}
\begin{proof}
Let $N_i$ be as in Lemma~\ref{lem:min-code}, so that $\mathcal{M}=\{\sigma \cup N_1, \sigma \cup N_2, \dots,\sigma \cup N_m \}$ is the set of facets of $\Delta$. Also, let $\tau_i=N_i \cap N_{i+1}$ for $1 \leq i \leq m-1$.  Then, Lemma~\ref{lem:min-code} implies that $\CC_{\rm min} (\Delta) = \{ \sigma \cup \tau_i \mid 1 \leq i \leq m-1\} \cup \mathcal{M} \cup \{ \emptyset \}$.  
We show that this code is convex with minimal embedding dimension 1 via a convex realization in $\mathbb{R}$ so that the region for $\sigma \cup N_1$ is the interval $(0,1)$, the $\sigma \cup \tau_1$ region is $[1,2]$, the $\sigma \cup N_2$ region is $(2,3)$, the region for $\sigma \cup \tau_2$ is $[3,4]$, and so on.  More precisely, the receptive fields $U_j$ for each neuron $j$ are:

\begin{equation*}
U_j ~=~ 
\begin{cases}
	(0,~2m-1) & \text{ if~} j \in \sigma \\
	(2i-2,~2i+1) & \text{ if~} j \in \tau_i, \text{~for~some~} 1 \leq i \leq m-1 \\
	(2i-2,~2i-1) & \text{ if~} j \in \left( N_i \setminus \left( \bigcup_{k=1}^{m-1}\tau_k \right) \right),  \text{~for~some~} 1 \leq i \leq m~.
\end{cases}
\end{equation*}
It is straightforward to check that this code has the regions described above.

\end{proof}

\begin{example}
Let $\Delta$ be the cone, with cone point $1$, over the following simplicial complex:
\begin{center}
\begin{tikzpicture}[scale=.8]
	\draw (0,0) --(1,0);
	\draw (2,0) --(3,0);
    \draw [fill=gray, thick] (1,0) --(2,0)--(1.5,.8) -- (1,0);
    \draw [fill] (0,0) circle [radius=0.08];
    \draw [fill] (1,0) circle [radius=0.08];
    \draw [fill] (2,0) circle [radius=0.08];
    \draw [fill] (3,0) circle [radius=0.08];
    \draw [fill] (1.5,.8) circle [radius=0.08];
    \node [below] at (0,0) {$2$};
    \node [below] at (1,0) {$3$};
    \node [below] at (2,0) {$4$};
    \node [above] at (1.5,.8) {$5$};
    \node [below] at (3,0) {$6$};
	\end{tikzpicture}
\end{center}
Then $\Delta$ is a coned path; indeed, the intersection of all facets of $\Delta$ is $\{1\}$, and the nerve $\mathcal{N}(\mathcal{L}_{\Delta}(\sigma))$ is the following path of length 2:
\begin{center}
\begin{tikzpicture}[scale=.9]
	\draw (0,0) --(1,0);
	\draw (1,0) --(2,0);
    \draw [fill] (0,0) circle [radius=0.08];
    \draw [fill] (1,0) circle [radius=0.08];
    \draw [fill] (2,0) circle [radius=0.08];
    \node [below] at (0,0) {$23$};
    \node [below] at (1,0) {$345$};
    \node [below] at (2,0) {$56$};
	\end{tikzpicture}
\end{center}
Following the proof of Proposition~\ref{prop:path-cvx}, the minimal code is $\CC_{\rm min} (\Delta) = \{ {\bf 123,~1345,~156},~13,~15,~\emptyset\}$, and a 1-dimensional convex realization is as follows, where the open intervals $U_i$ are depicted above the real line for clarity: 
\begin{center}
\begin{tikzpicture}[scale=.7]
    \draw (-2,0) -- (7,0);
    \node [right] at (7,0) {$\mathbb{R}$};
	\draw (4,0.5) --(5,0.5);
	\node [left] at (4,0.5) {$U_6$};
	\draw (2,1) --(5,1);
	\node [left] at (2,1) {$U_5$};
	\draw (2,1.5) --(3,1.5);
	\node [left] at (2,1.5) {$U_4$};
	\draw (0,2) --(3,2);
	\node [left] at (0,2) {$U_3$};
	\draw (0,2.5) --(1,2.5);
	\node [left] at (0,2.5) {$U_2$};
	\draw (0,3) --(5,3);
	\node [left] at (0,3) {$U_1$};
\end{tikzpicture}
\end{center}
From left to right, the nonempty codewords are $123$, $13$, $1345$, $15$, and $156$.
\end{example}

\section{Discussion} \label{sec:openQ}

We resolved the problem of whether all neural codes with no local obstructions are convex.  
This motivates some related questions.  
First, are all neural codes with no local obstructions good-cover codes?  Even for codes on five neurons, this question is unresolved. Our enumeration of codes on five neurons without local obstructions is a step toward attacking these problems.

Next, are all max-intersection-complete codes (i.e.\ closed under maximal intersection) convex?  
We posed this question in an earlier version of this work, and since then 
Cruz {\em et al.} have answered it in the affirmative \cite{intersection-complete}.

The next question arises from codes of the form $\CC_{\rm min} (\Delta)$, which we recall is the smallest neural code $\mathcal{C}$ with no local obstructions such that $\Delta(\CC)=\Delta$. (We could also consider a minimal convex code with a given simplicial complex, although this is not, in general, unique: recall the codes discussed in Remark~\ref{rmk:2-minimal}.) 
If the minimal code on a simplicial complex is convex, are all codes on the same simplicial complex that contain the minimal code convex? 
Again, Cruz {\em et al.}~\cite{intersection-complete} recently answered this question in the affirmative.  Accordingly, this reduces the classification of convex codes to the determination of the minimal convex codes on each simplicial complex. 

Finally, we pose some questions that arise from our counterexample code.  First, our proof that this code is non-convex hinged upon a forced two-dimensional structure. Similar techniques have helped us to draw receptive fields for other five-neuron codes which permitted exclusion of an intersection of maximal codewords. Can we use these techniques to draw place fields in general, or to determine bounds on minimal embedding dimension? 

Additionally, can we characterize a new obstruction such that neural codes are convex if and only if they do not have this new obstruction?  Can we give an algebraic signature for such an obstruction (see~\cite{what-makes})?  Whatever the signature of our new obstruction is, it is clear that it cannot have as simple an interpretation as the local obstruction. While a local obstruction specifies a set of missing codewords, all of which must be added to the code to resolve the obstruction, the obstruction to convexity present in our counterexample code 
has a more complicated structure:
adding \emph{either} the codeword 1 or both the codewords 234 and 345 makes the code convex. One form this new obstruction could take is the requirement that convex codes either be max-intersection-complete or satisfy some other condition, yet to be determined.

\subsection*{Acknowledgments}
CL and ZW conducted this research as part of the NSF-funded REU in the Department of Mathematics at Texas A\&M University (DMS-1460766), in which AS served as mentor.
The authors benefited from guidance from Lauren Grimley and Jacob White, and from discussions with Carina Curto, Bryan F\'elix, Chad Giusti, Elizabeth Gross, Vladimir Itskov, Katie Morrison, William Kronholm, Sean Owen, and Nora Youngs.
The authors also thank Joseph Kung for editorial suggestions.  
AS was supported by the NSF (DMS-1312473/DMS-1513364).


\bibliographystyle{plain}
\bibliography{convexity-refs}

\section*{Appendix: Classification of codes on five neurons with no local obstructions}

Below we list (up to isomorphism) all 157 connected simplicial complexes on five vertices; here the simplicial complexes are listed by their facets, and the vertex set is $\{0,1,2,3,4\}$.  For each simplicial complex, we list the mandatory codes, the non-mandatory intersections of facets, the number of non-mandatory codewords, and the number of codes with no local obstructions.


\includepdf[pages=-]{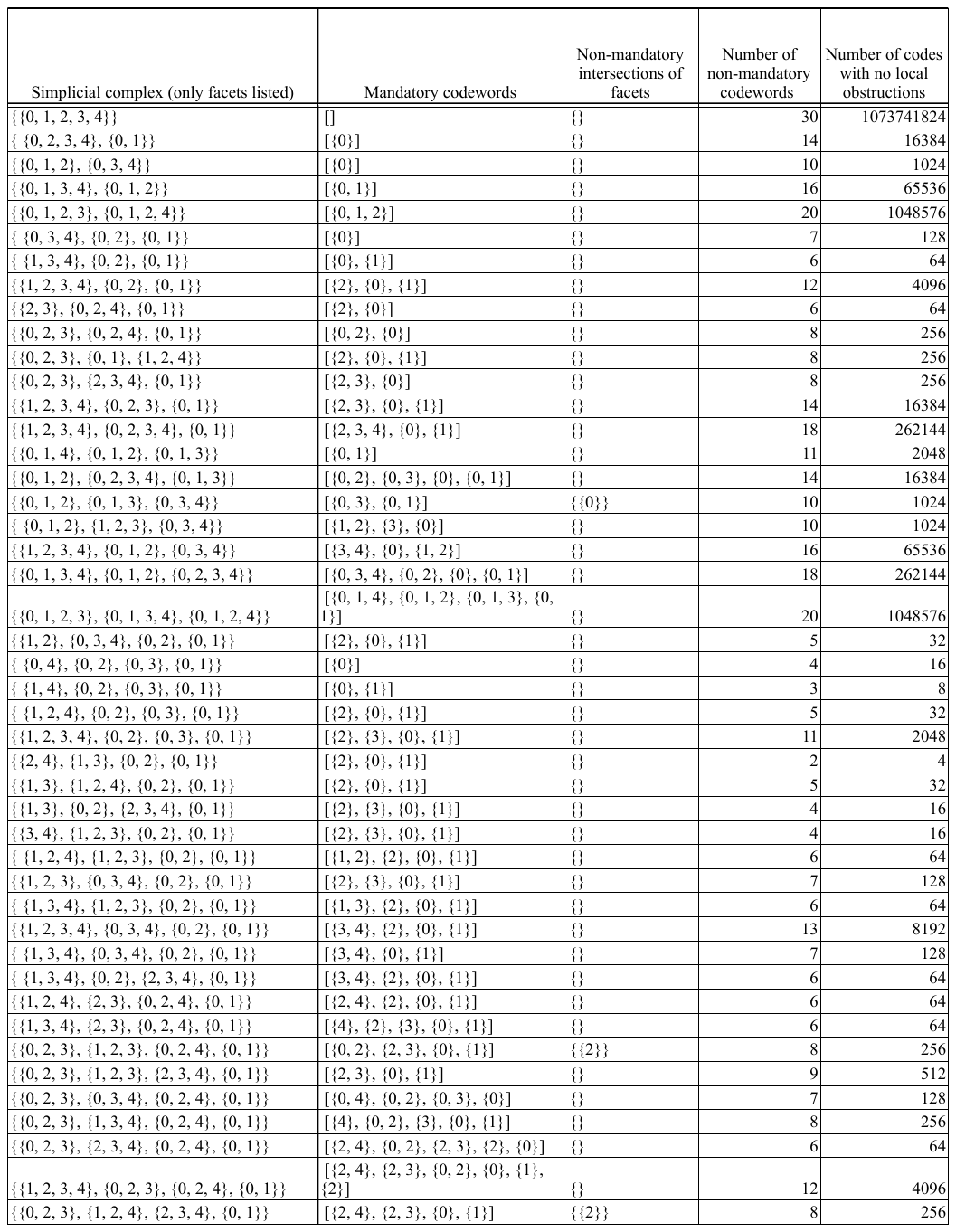}

\end{document}